\documentclass{l4dc2024}

\title[Learning for Sys-ID of NDAE-modeled Power Systems]{Learning for System Identification of NDAE-modeled Power Systems}
\usepackage{times}

\usepackage{comment}

\usepackage{tikz}

\newtheorem{asmp}{Assumption}

\newtheorem{myprs}{Proposition}

\author{%
 \Name{Wenjie Mei} \Email{wenjie.mei@vanderbilt.edu}\\
 \addr Department of Civil and Environmental Engineering, Vanderbilt University, 2201 West End Avenue, Nashville, TN 37235, USA
 \AND
 \Name{Muhammad Nadeem} \Email{muhammad.nadeem@vanderbilt.edu}\\
 \addr Department of Civil and Environmental Engineering, Vanderbilt University, 2201 West End Avenue, Nashville, TN 37235, USA
 \AND
 \Name{MirSaleh Bahavarnia} \Email{mirsaleh.bahavarnia@vanderbilt.edu}\\
 \addr Department of Civil and Environmental Engineering, Vanderbilt University, 2201 West End Avenue, Nashville, TN 37235, USA
 \AND
 \Name{Ahmad F. Taha} \Email{ahmad.taha@vanderbilt.edu}\\
  \addr Department of Civil and Environmental Engineering, Department of Electrical and Computer Engineering, Vanderbilt University, 2201 West End Avenue, Nashville, TN 37235, USA
}

\begin{document}

\maketitle

\begin{abstract}%
System identification through learning approaches is emerging as a promising strategy for understanding and simulating dynamical systems, which nevertheless faces considerable difficulty when confronted with power systems modeled by differential-algebraic equations (DAEs). This paper introduces a neural network (NN) framework for effectively learning and simulating solution trajectories of DAEs. The proposed framework leverages the synergy between Implicit Runge-Kutta (IRK) time-stepping schemes tailored for DAEs and NNs (including a differential NN (DNN)). The framework enforces an NN to cooperate with the algebraic equation of DAEs as hard constraints and is suitable for the identification of the ordinary differential equation (ODE)-modeled dynamic equation of DAEs using an existing penalty-based algorithm. Finally, the paper demonstrates the efficacy and precision of the proposed NN through the identification and simulation of solution trajectories for the considered DAE-modeled power system.
\end{abstract}


\section{Introduction}
System identification for power networks is essential in the domain of power systems for many reasons: First, accurate models are necessary for effective design, analysis, and control, offering insights into system behaviors. Identifying model parameters allows engineers to further study power dynamics performance and stability, for example. Second, identification can produce predictive models for planning or optimizing the behaviors of power systems, which is important in advanced power models with renewables and emerging technologies. The identified power models obtained through identification techniques facilitate control strategies for enhancing the efficiency and resilience of the controlled systems. 

Differential-algebraic equations (DAEs) are essential in power systems due to their ability to model dynamic behaviors and interconnections. They offer a solid mathematical framework capturing differential equations for dynamic responses and algebraic constraints between causal components. For instance, DAEs enable a well-formed representation of interconnected power system components, facilitating the modeling of generators and transformers. This advantage is beneficial as power grids evolve with renewable energy and advanced integration. 

\subsection{Existing works}
The literature pays considerably less attention to system identification challenges in DAEs compared to, for instance, nonlinear ODE models (refer to \cite{schon2015sequential} and \cite{sjoberg1995nonlinear} for more details). Existing methods for DAEs are primarily designed for scenarios of addressing numerical aspects within the optimization problem and involving environments with or without disturbances, under the condition that all inputs are known. In such cases, these methods typically have straightforward formulations, with a focus on addressing numerical aspects within the optimization problem (see, for instance, \cite{esposito2000global} and \cite{bock2007numerical}, along with relevant references). Also, there exist some works studying the identification of DAEs with process disturbances, \emph{e.g.}, in \cite{abdalmoaty2021identification}. However, the modeling of DAEs frequently encounters challenges in accurately representing nonlinearities through traditional methods. Neural networks (NNs), famous for their usefulness in capturing complex nonlinear interconnections, exhibit high performance, indicating they are ideal for accurately representing the intricate dynamics in DAE-modeled power systems. 

NNs have extensively been investigated for learning dynamical systems that are described by ordinary differential equations (ODEs), considerably in nonlinear ODEs (NODEs), providing us with
an effective alternative to conventional numerical methods with high costs \cite{lu2021deepxde,yazdani2020systems,meade1994solution}.
Although NNs have succeeded in learning solution trajectories or simulating dynamic behaviors for ODEs \cite{chang2019antisymmetricrnn,fang1996stability}, developing an NN-based framework for learning and simulating solution trajectories in the context of nonlinear DAEs (NDAEs) remains as an unresolved challenge. The difficulty arises from the fact that DAEs exhibit infinite stiffness property \cite{knorrenschild1992differential,kim2021stiff}, which can result in gradient pathologies \cite{wang2021understanding} and optimization problems with poor conditioning. These issues often lead to the failure of many trainings. To that end, observing the general forms of NDAEs (especially the dynamic system), in this paper, we introduce a differential NN (DNN) \cite{chen2018neural}, which is a continuous variant of an artificial NN (ANN). DNNs incorporate a feedback element, imparting a memory effect, enabling them to operate with the historical data of a process---a feature not feasible for static algorithms. The significant attribute of a DNN lies in its ability to identify unknown systems by incorporating established ODEs.

\subsection{Main Contributions}

The main contributions of the current paper can be summarized as follows:

\begin{itemize}
    \item We present an NN scheme for identifying the parameters of an NDAE-modeled power system in which a DNN is included to approximate the dynamic equation of the system. Our approach can also be well-adapted to the general NDAEs with an index of $1$, which is the inherent physical property of the considered power model.
    \item To guarantee that the proposed NN architecture is effective, under some standard assumptions (imposed on nonlinearities), we further analyze the error dynamic between the true system and the proposed DNN and give the conditions for bounding the identification error at the infinite-time horizon. This is demonstrated by a numerical simulation. We also illustrate that the imposed conditions for identification error analysis are milder than the existing ones in the literature.
\end{itemize}

%Moreover, DAEs play a crucial role in simulation and analysis, aiding in the prediction and evaluation of power system performance under various conditions. This capability is vital for system design, optimization, and the development of robust control strategies to enhance power grid stability and reliability.

\subsection{Paper organization}

The rest of this paper is organized as follows: Section~\ref{sec:problem_state} presents the considered power system and an overview of the system identification problem. In Section~\ref{sec:pro_method}, we first give the details of the proposed NN scheme and the numerical methods for training and testing and then present the boundedness conditions for the identification errors. To illustrate the efficacy of the identification method, we present a case study of the considered DAE-modeled power system in Section~\ref{sec:Example}. The paper summary and future research directions in Section~\ref{sec:conclusion} conclude the paper. The used notation is provided next:
%\vspace{-1em}

\subsection*{Notation}
The set of real numbers (respectively non-negative real numbers) is
  denoted by $\mathbb{R}$ (respectively $\mathbb{R}_{\geq 0}$). The symbol $\rVert \cdot \rVert$ denotes the Euclidean norm on
  $\mathbb{R}^{n}$ (and the induced matrix norm $\rVert A \rVert$ for a matrix
  $A\in\mathbb{R}^{m\times n}$). We use 
  $I_{n}$ to denote the $n\times n$
  identity matrix. The minimum eigenvalue of a symmetric matrix $P$ is denoted by $\lambda_{\min}(P)$. To represent the positive definiteness and negative semi-definiteness, we utilize $\succ 0$ and $\preceq 0$, respectively. For $t_{1}$, $t_{2}\in \mathbb{R}$, with $t_{1}<t_{2}$, we
  denote by $C^{1}_{n}([t_{1},t_{2}])$ the Banach space of continuously differentiable
  functions $\psi:[t_{1},t_{2}] \to\mathbb{R}^{n}$ with the norm
  $\Vert\psi\Vert_{[t_{1},t_{2}]}=\sup_{ r \in [t_1,t_2]} \|\psi(r)\| + \sup_{ r \in [t_1,t_2]}\|\frac{d \psi(r)}{d r}\| < +\infty$. For brevity, we denote a time-dependent signal $s(t)$ by $s$ wherever needed (\emph{e.g.}, $x_d$ instead of $x_d(t)$).

\section{Problem Statement} \label{sec:problem_state}

In this paper, we consider the system identification problem for the standard $4^{\text{th}}$-order power system model \cite{sauer2017power, nadeem2022dynamic}, which is represented via the following set of equations:
% \begin{subequations}
% 		\begin{align}
% 		\textit{nonlinear generator ODEs} \;\;\;\;\;	\dot{{ x}}_d &= { A}_d{x}_d +  {f}_d\left({\m x}_d,{\m x}_a\right) + {\m B}_d \textcolor{blue}{\m u}\\
% 		\textit{nonlinear power flow} \;\;\;\;\;\;	\m 0 &= {\m A}_a{\m x_a} + {\m f}_a\left({\m x}_d,{\m x}_a\right) + {\m {B}_{a}} \textcolor{red}{{{\m w}} }
% 		\end{align}
% \end{subequations}
\begin{subequations} \label{eq:main_power_sys}
\begin{align}
\begin{split} \label{main_DAE_system_1}
\textit{Generator ODEs}: \;\;\;\;\; E_d \dot{x}_d &= A_d x_d + C_d f(x_d,x_a) +  B u + hw_0, \end{split}\\
\begin{split}  \label{main_DAE_system_2}
\textit{Power flow equations}: \;\;\;\;\;\; 0 & = A_a x_a + C_a g(x_d,x_a),
\end{split}
\end{align}
\end{subequations} 
where $x_d \in\mathbb{R}^{n_d}$ contains the dynamic states, $x_a \in\mathbb{R}^{n_a}$ contains the algebraic variables, $A_d$, $C_d$, $B$, $A_a$, and $C_a$ are constant system matrices with appropriate dimensions, $u \in \mathbb{R}^m$ contains the control inputs while the functions $f: \mathbb{R}^{n_d} \times \mathbb{R}^{n_a} \to \mathbb{R}^{n_{f}}$ and $g: \mathbb{R}^{n_d} \times \mathbb{R}^{n_a} \to \mathbb{R}^{n_{g}}$ (it is continuously differentiable) are the nonlinearities in dynamic and algebraic equations, respectively; $E_d = I_{n_d}$ holds and $w_0 \in \mathbb{R}$ is the synchronous speed; the vector $h$ satisfies $h \in \mathbb{R}^{n_d}$. In the above power system model, Eq.~\eqref{main_DAE_system_1} lumps the dynamic equations (set of ODEs) for all the synchronous generators while~\eqref{main_DAE_system_2} includes the power flow/balance equations (the algebraic constraints) of the power grid. Due to space limitations, the detailed description/explanation of the power system model is not included in this work and can be found in the power system literature, \emph{e.g.}, see the references \cite{nadeem2022dynamic, qi2018comparing, sauer2017power}. 

That being said, \emph{w.l.o.g.}, we let $f(0,0)=0$ and $g(0,0)=0$ and assume that the DAE~\eqref{eq:main_power_sys} is with index $1$ (this is a physical property of the power system~\eqref{eq:main_power_sys}), which means that the Jacobian $\frac{\partial \tilde{g}}{\partial x_a}$ is invertible for $\tilde{g}(x_d,x_a) :=  A_a x_a + C_a g(x_d,x_a)$. This implies that, by Implicit Function Theorem \cite{krantz2002implicit}, for the algebraic Eq.~\eqref{main_DAE_system_2}, there exists a unique solution: $x_a = \ell(x_d)$, where $\ell$ is a function. If we substitute $\ell(x_d)$ into $x_a$ in~\eqref{main_DAE_system_1}, we then get $E_d \dot{x}_d = \dot{x}_d = A_d x_d + C_d f(x_d,\ell(x_d)) +  B u + hw_0$, which is a standard NODE, and the corresponding system identification techniques for NODEs can be applied in this case.

Throughout the paper, we deal with the system identification problem for~\eqref{eq:main_power_sys}, specifically, the \textit{\textbf{black-box}} identification. To collect data sets for identifications, we also consider the output of the NDAE system~\eqref{eq:main_power_sys} as follows: 
\begin{align} \label{eq:output_data}
    y &= \begin{bmatrix}
    x_d^T & x_a^T
\end{bmatrix}^T. 
\end{align}
For system identification, the model~\eqref{eq:main_power_sys} is further assumed to be smoothly parametrized by a
constant unknown parameter vector $\theta \in \mathbb{R}^{n_\theta}$, which is to be estimated.  Notice that in the implementation of identification, the input $u(t)$ is known and given by the human operator, while the output $y(t)$ has to be sampled at each time instant denoted by $t_k$. Then, we have the following complete data collection: 
\begin{align*}
N_{\eta,T} &:= \{ (y(t_{k}), u(s)) \mid k=1,\dots,\eta, \; s \in [0,T] \},
\end{align*}
for some $T>0$. The main goal of this work is to estimate the parameter $\theta$ by using the data set $N_{\eta,T}$. To that end, we first propose an NN architecture for identifying the dynamic and algebraic models and then analyze the corresponding identification error between the true output and the predicted output generated by the proposed NN.  
%\section{NDAE Model} \label{sec:NDAE_model}
%This section...
\section{Proposed Method} \label{sec:pro_method}
This section presents the proposed method for identifying the dynamic and algebraic equations of the NDAE~\eqref{eq:main_power_sys}. For the ODE-modeled dynamic system, a DNN scheme is employed, while for the algebraic one, a standard feedforward neural network (FNN) is utilized. % First, the implicit Runge-Kutta approach is utilized to obtain data sets, which are used for training or testing the models. Then, we give the proposed architecture of NNs. Thirdly, the estimation method is shown to demonstrate the details of implementing identifications. Finally, an identification error analysis is performed to illustrate the performance of DNN in identifying the dynamic equation. 
\subsection{Implicit Runge-Kutta (IRK) method}
In the system identification of~\eqref{eq:main_power_sys}, we assume that the data of $\{x_d(t_k),x_a(t_k)\}_{k=1}^\eta$ is accessible to the human operator in discrete time. Also, numerical analysis methods can be used to approximate solutions of~\eqref{eq:main_power_sys}. In the current paper, we utilize the Implicit Runge-Kutta (IRK) method with $\nu$ stages expressed as follows:
\begin{subequations} \label{eq:Runge-Kutta}
\begin{align}
\begin{split} \label{Runge-Kutta_1}
  \alpha_j&= x_d^n + \Delta \sum_{i=1}^\nu b_{j,i} \tilde{f}\left(\alpha_i, \beta_i, u\right), \quad  0  = \tilde{g}(\alpha_j, \; \beta_j), \end{split}  \\
\begin{split}  \label{Runge-Kutta_3}
x_d^{n+1} &= x_d^n + \Delta \sum_{j=1}^\nu c_{j} \tilde{f} \left(\alpha_j, \beta_j,  u\right), ~\quad 0  = \tilde{g}(x_d^{n+1},\; x_a^{n+1}),
\end{split} 
\end{align}
\end{subequations} 
where $j = 1,\dots,\nu$; $\Delta>0$ is sufficiently small, $x_d^{n} :=x_d(t_{n})$, 
\begin{align*}
\alpha_j &:= x_d\left(t_n+ \Delta \sum_{k=1}^\nu b_{j,k}\right),  \quad  \beta_j := x_a\left(t_n+ \Delta \sum_{k=1}^\nu b_{j,k}\right),  \\ \tilde{f}(x_d,x_a,u) &:= A_d x_d + C_d f(x_d,x_a) +  B u + hw_0, \quad
 \tilde{g}(x_d,x_a) := A_a x_a + C_a g(x_d,x_a).
\end{align*}
Note that $\{b_{j,i},c_j\}$ are known parameters and since $E_d =  I_{n_d}$ holds, $E_d$ is omitted.  

For the IRK scheme~\eqref{eq:Runge-Kutta}, after parametrization (with the denotation $\theta$; for example, $\alpha_j$ is changed to $\alpha_j^\theta$), we can denote the dynamic states and algebraic variables by
\begin{align*}
[\alpha_1^\theta,\; \dots, \; \alpha_\nu^\theta, \; x_d^{n+1,\theta}],\quad
[\beta_1^\theta,\; \dots, \; \beta_\nu^\theta, \; x_a^{n+1,\theta}].
\end{align*}
These can be used in the training or testing process. 

\subsection{Architecture of neural networks}
\begin{figure}[h!]
\centering
\includegraphics[width=5.5in]{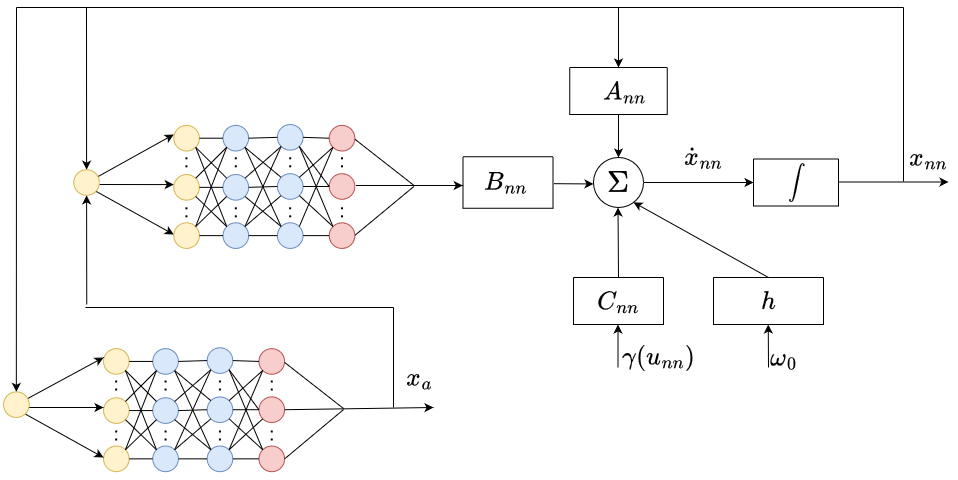}
\caption{A schematic design of the proposed system identification technique.}
\label{Fig:NN_architecture}
\end{figure}

In this section, we consider an architecture of NNs for identifying the DAE system~\eqref{eq:main_power_sys}, in which the DNN is for identifying the dynamic Eq.~\eqref{main_DAE_system_1} and another FNN is for approximating the algebraic constraint~\eqref{main_DAE_system_2}. This scheme is depicted in Fig.~\ref{Fig:NN_architecture}, in which the expression of the DNN is: $ \dot{x}_{nn} = A_{nn} x_{nn} + B_{nn} \hat{\rho}(x_{nn}) + h \omega_0 + C_{nn} \gamma(u_{nn})$, where the function $\hat{\rho}$ represents the upper NN; the matrices $A_{nn}, B_{nn}, C_{nn}$, the operator $\gamma$, the input $u_{nn}$ are defined as in \eqref{eq:DNN}, and $x_{nn}$ is the state of the DNN (it approximates $x_d$). Also, in the identification processes, the matrices or functions $A_{nn}, B_{nn}, C_{nn}, \gamma, u_{nn}, h, \omega_0$ are assumed to be known.

\subsection{Estimation method}
In the identification process, the two constraints have to be satisfied (\emph{i.e.}, the NNs are restricted to satisfy the differential and algebraic equations described by~\eqref{eq:Runge-Kutta}). At the beginning, we select the initial conditions of NNs in the whole considered domain (the initial conditions are randomly distributed). Also, the output training or testing data (\emph{i.e.}, $x_d^n, x_a^n$) are measured by the human operator so that we have two training sets of size $\eta>0$ for the two NNs (one for the dynamic system~\eqref{main_DAE_system_1}, and another one for the algebraic Eq.~\eqref{main_DAE_system_2}), respectively, that is
\begin{align*}
D_{\eta}^d := \{ x_d^{n,1},  \dots, x_d^{n,\eta} \},\quad D_{\eta}^a := \{ x_a^{n,1},  \dots, x_a^{n,\eta} \}.
\end{align*}
%These two data sets are obtained from the output data of~\eqref{eq:output_data}.

Then, to analyze the difference between the NNs and the IRK scheme dynamics~\eqref{eq:Runge-Kutta}, we consider the following loss function:
\begin{equation*}
    \mathcal{L}({\theta,\eta}) :=  l(\mathcal{L}_d({\theta,\eta}), \mathcal{L}_a({\theta,\eta})),
\end{equation*}
where $l:\mathbb{R}_{\geq 0} \times \mathbb{R}_{\geq 0} \to \mathbb{R}_{\geq 0}$ is an operator with $l(0,0)=0$, and the functions $l(s,\cdot), l(\cdot,s)$ increase as the variable $s \in \mathbb{R}_{\geq 0}$ increases, where $\cdot $ can be any fixed non-negative real number.

In this paper, $\mathcal{L}_d({\theta,\eta})$ and $\mathcal{L}_a({\theta,\eta})$ are specifically defined as
\begin{align*}
    \mathcal{L}_d({\theta,\eta}) &:=  \frac{1}{\eta (\nu+1) } \sum_{x_d^n \in D^d_{\eta}} \; \sum_{j=1}^{\nu+1} \left \| x_d^n - \hat{x}_d^{n,j}(\theta) \right\|^2, \\
\mathcal{L}_a({\theta,\eta}) &:=  \frac{1}{ \nu+1}  \tilde{\ell} \left(\sum_{j=1}^{\nu}  \| \tilde{g}(\alpha_j^\theta,\beta_j^\theta) \| + \left( \| \tilde{g}(x_d^{n+1,\theta},x_a^{n+1,\theta}) \| \right) \right), 
    \end{align*}
where 
\begin{align*}
& \hat{x}_d^{n,j}(\theta) := \alpha_j^\theta - h\sum_{i=1}^\nu b_{j,i} \tilde{f}\left(\alpha_i^\theta, \beta_i^\theta, u \right), \quad 
\hat{x}_d^{n,\nu+1}(\theta) := x_d^{n+1,\theta} - h \sum_{j=1}^\nu c_{j} \tilde{f}(\alpha_j^\theta, \beta_j^\theta, u), 
\end{align*}
and the function $\tilde{\ell}(s)$ increases as the variable $s \in \mathbb{R}_{\geq 0}$ increases. 

Next, we use the two NNs (as depicted in Fig.~\ref{Fig:NN_architecture}) to identify the dynamic and algebraic equations. We train the parameters of the considered NNs by minimizing the loss function using a gradient-based optimizer \cite{daoud2023gradient,dogo2018comparative}, which can be expressed as
\begin{align*}
\hat{\theta} &= \mathop{\arg  \min_{\theta} } \; \mathcal{L} ({\theta,\eta}). 
\end{align*}
Note that the selection of the explicit form of the operator $\mathcal{L}$ can be tricky since there are two main constraints originating from the DAE system~\eqref{eq:main_power_sys}, \emph{i.e.}, the solution of~\eqref{eq:main_power_sys} should satisfy the dynamic system~\eqref{main_DAE_system_1} and stay within the manifold defined as follows: 
\begin{align*}
M &:= \{ (x_d,x_a) \mid \tilde{g}(x_d,x_a) = 0 \}, 
\end{align*}
which is deduced from~\eqref{main_DAE_system_2}. For that purpose, we select a recurrent penalty method in the sequel to get a well-conditioned loss function. Also, for simplicity, we select the loss function as 
\begin{align*}
    \mathcal{L}({\theta,\eta}) &=  l(\mathcal{L}_d({\theta,\eta}), \mathcal{L}_a({\theta,\eta})) := w_d \mathcal{L}_d({\theta,\eta}) + w_a \mathcal{L}_a({\theta,\eta}).
\end{align*}
It is noteworthy that we employ an existing algorithm in the literature (Algorithm $2.1$ in \cite{lu2021physics} or Algorithm $1$ in \cite{moya2023dae}) for determining the values of $w_d, w_a >0$. Due to space limitations, we are unable to provide the details of such an algorithm in the current paper. 
%\begin{algorithm} \caption{A penalty method for determining loss function }\label{alg:penalty_method} \begin{algorithmic}[1] \State Initiate the values of $w_d^0, w_a^0$, factor $\xi$, and iteration number $K$ \State $k \leftarrow 0$, and  $\theta^0 \leftarrow \mathop{\arg \; \min_{\theta} }  \mathcal{L}^0({\theta,\eta})$ \State ... \end{algorithmic} \end{algorithm}

\subsection{Identification error analysis of DNN}
In this section, we consider the following structure of general DNN (the DNN shown in Fig.~\ref{Fig:NN_architecture} is a special case) for identifying~\eqref{main_DAE_system_1}:
\begin{align} \label{eq:DNN}
    \dot{x}_{nn} &= A_{nn} x_{nn} + B_{nn} \rho(x_{nn}) + C_{nn} \gamma(u_{nn}), 
\end{align}
where $A_{nn} \in \mathbb{R}^{n \times n} $ is the state matrix of the DNN; $B_{nn} \in \mathbb{R}^{n \times n_B}, C_{nn} \in \mathbb{R}^{n \times n_C} $ are also matrices; the function $\rho: \mathbb{R}^{n} \to \mathbb{R}^{n_B}$ is nonlinear, the operator $\gamma: \mathbb{R}^{n_u} \to \mathbb{R}^{n_C}$, and the input $u_{nn} := u$ (thus, $n_u = m$), here $u$ contains the known inputs given in~\eqref{eq:main_power_sys}. The subscript $``nn"$ represents the NN. Then, we analyze the identification error after the NN
parameters have been trained using the previous
learning method. 

In the identification, one can assume the system state can be measured, and the dimensions of the input and states (\emph{i.e.}, $x_d, x_a$) of the true system are equal to those of the DNN. Under these mild assumptions, we then study the dynamic behavior of the error 
\begin{align*}
   e &:= x_d - x_{nn}. 
\end{align*}
The resulting error dynamics can be expressed as follows:
\begin{align} \label{eq:error_system}
    \dot{e} &= Ae + \tilde{\phi}(e,x_d,x_a,u,t), 
\end{align}
where $A$ is a Hurwitz matrix, the function $\tilde{\phi}(e,x_d,x_a,u,t) := A_d x_d + C_d f(x_d,x_a) +  B u + hw_0 -  A_{nn} x_{nn} - B_{nn} \rho(x_{nn}) - C_{nn} \gamma(u_{nn}) - Ae$. In the expression of $\tilde{\phi}(e,x_d,x_a,u,t)$, we have the variable $x_a$, which is estimated by the previously-proposed FNN (see Fig.~\ref{Fig:NN_architecture} for clarification). Furthermore, there exists a function $\ell$ such that $x_a = \ell(x_d)$ holds under the condition that the function $\tilde{g} $ is continuously differentiable and its Jacobian is invertible (or the DAE system~\eqref{eq:main_power_sys} is with index $1$). Then, the variable $x_a$ in $\tilde{\phi}(e,x_d,x_a,u,t)$ can be substituted by $\ell(x_d)$, thus obtaining 
\begin{align}\label{eq:error_system_replace}
\phi(e,x_d,u,t) := ~& A_d x_d + C_d f(x_d,\ell(x_d)) +  B u + hw_0 \nonumber \\
& -A_{nn} x_d - (A-A_{nn})e - B_{nn} \rho(x_d-e) - C_{nn} \gamma(u).
\end{align}
For the system~\eqref{eq:error_system_replace}, several assumptions are imposed in this work to guarantee that the error $e$ is bounded. We first introduce the following assumption on the nonlinearity $\phi$: 
\begin{asmp} \label{assum_1}
    There exist matrices $L \succ 0$ and $K \succ 0$ such that the following inequality:
    \begin{align*}
        & \phi(e,x_d,u,t)^\top  L \phi(e,x_d,u,t) \leq c_0(x_d,u,t)  + c_1(x_d,u,t) e^\top K e,
    \end{align*}
    holds where the real numbers $c_0, c_1 >0$ are bounded by
    \begin{align*}
        & \sup \; c_0(x_d,u,t) = c^0 < +\infty, \quad \sup \; c_1(x_d,u,t) = c^1 < +\infty,  \\ 
    & \forall x_d \in C^1_{n_d}(\mathbb{R}_{\geq 0}), \;
    u \in C^1_m(\mathbb{R}_{\geq 0}), \; t \geq 0. 
    \end{align*}
\end{asmp}

Note that Assumption~\ref{assum_1} is non-restrictive, since the right-hand side terms of~\eqref{eq:error_system_replace} can respectively be upper bounded by
    \begin{align*}
& (A_d -  A_{nn}) x_d + C_d f(x_d,\ell(x_d)) +  B u + hw_0- B_{nn} \rho(x_d) - C_{nn} \gamma(u)  \leq \; c_0(x_d,u,t), \\
&- (A-A_{nn})e + \left(B_{nn} \rho(x_d)- B_{nn} \rho(x_d-e)  \right) \leq \; c_1(x_d,u,t) e^\top K e, \\ 
& \forall x_d \in C^1_{n_d}(\mathbb{R}_{\geq 0}), \;
   u \in C^1_m(\mathbb{R}_{\geq 0}),  \; t \geq 0, 
\end{align*}
    under the conditions that $\rho$ is a Lipschitz function, the norm of $\gamma(t)$ at $t \geq 0$ is finite, and $f(s(t),\ell(s(t)) $ $< +\infty$ holds for $\| s \|_{\mathbb{R}} < +\infty$. Therefore, we then impose the second assumption on the nonlinearity $\rho$ and the operators $\gamma, f$ as follows:  
\begin{asmp} \label{assum_2}
    The nonlinearity $\rho$  is a Lipschitz function, the operator $\gamma$ satisfies $\| \gamma \|_{\mathbb{R}} < +\infty$, and $f(s(t),\ell(s(t)) < +\infty$ holds for $\| s \|_{\mathbb{R}} < +\infty$. 
\end{asmp}
The condition on the norm $\| \cdot \|_{\mathbb{R}}$ is more restrictive than $\| \cdot \|_{[r^1,r^2]}$, where $r^1:= \inf_{t \geq 0}  r(t), \; r^2:= \sup_{t \geq 0}   r(t)$, where $r(t)$ is a scalar-valued function. The latter case can also be considered for further relaxation for the system runs on the finite time interval $[r^1,r^2]$. 

We are now ready to give the third assumption: 
\begin{asmp} \label{assum_3}
    Assume that there exist matrices $P = P^\top \succ 0, W = W^\top \succ 0$ such that for the system~\eqref{eq:error_system}, the following matrix inequality:
  \[
  A^\top P+ PA + P L^{-1} P + c^1 K  +  W \preceq 0, 
  \]
  has a solution $(P,W)$, where $L,c^1,K$ can be found in Assumption~\ref{assum_1}.
\end{asmp}
Assumptions~\ref{assum_1}-\ref{assum_3} are used in the proof of the following proposition: 
\begin{myprs} \label{thm:main}
Let Assumptions~\ref{assum_1}-\ref{assum_3} be satisfied. For the NDAE system~\eqref{eq:main_power_sys} and the DNN~\eqref{eq:DNN}, the error $e(t)$ satisfies the following property: 
\begin{align*}
& \lim_{t \to +\infty} \sup \| e(t) \| \leq \sqrt{\frac{c^0}{\lambda_{\min}(P) \lambda_{\min}(P^{-\frac{1}{2}} W P^{-\frac{1}{2}}) }}. 
\end{align*}
where $P^{-\frac{1}{2}}$ denotes the principal square root of $P^{-1} \succ 0$.
\end{myprs}

\begin{proof}
Consider the Lyapunov function
\begin{align*}
& V(e) = e^\top P e, \; \text{where} \; P = P^\top \succ 0. 
\end{align*}
Taking the time derivative of $V(e)$, we have
    \begin{align*}
& \dot{V}(e) = 2 e^\top P \dot{e} = 2 e^\top P \Big( Ae + \phi(e,x_d,u,t) \Big) = e^\top (A^\top P + PA)e + 2 e^\top P \phi .
  \end{align*}
Here, for the term $2 e^\top P \phi$, by Matrix Young Inequalities \cite{ando1995matrix}, we derive that for any $L = L^\top \succ 0$ (\emph{w.l.o.g.}, we select the matrix $L$ in Assumption~\ref{assum_1}), the following equality:
\begin{align*}
& 2 e^\top P \phi = (e^\top P) \phi + \phi^\top (P e) \leq \phi^\top L \phi+ e^\top P L^{-1} P e, 
\end{align*}
holds. Applying Assumption~\ref{assum_1}, we get 
\begin{align*}
& 2 e^\top P \phi \leq c_0(x_d,u,t)  + c_1(x_d,u,t) e^\top K e + e^\top P L^{-1} P e. 
\end{align*}
   Therefore, it can be deduced that
    \begin{align*}
\dot{V}(e) \leq ~& e^\top (A^\top P + PA)e + c_0(x_d,u,t)  + c_1(x_d,u,t) e^\top K e + e^\top P L^{-1} P e \\
\leq ~& e^\top \Big( A^\top P + PA + P L^{-1} P + c^1 K \Big) e +  c^0.
  \end{align*}
 Then, by Assumption~\ref{assum_3}, we have the following inequality: 
    \begin{align*}
\dot{V}(e)  \leq & - e^\top W e + c^0 
 \leq  - e^\top P^{\frac{1}{2}} \Big( P^{-\frac{1}{2}} W P^{-\frac{1}{2}} \Big) P^{\frac{1}{2}} e + c^0
 \leq  -\lambda_{\min}(P^{-\frac{1}{2}} W P^{-\frac{1}{2}})  V(e) + c^0. 
  \end{align*}
Let $\tilde{W}:= P^{-\frac{1}{2}} W P^{-\frac{1}{2}}$. We solve the above inequality and obtain  
    \begin{align} \label{eq:final_ineq}
\lambda_{\min}(P) \| e(t) \|^2  \leq ~& V(0) e^{-\lambda_{\min}(\tilde{W}) t} + \frac{c^0}{\lambda_{\min}(\tilde{W})} e^{-\lambda_{\min}(\tilde{W}) t} \Big( e^{\lambda_{\min}(\tilde{W}) t} -1 \Big)
 \leq \frac{c^0}{\lambda_{\min}(\tilde{W})},
  \end{align}
due to $0< e^{-\lambda_{\min}(\tilde{W}) t} \Big( e^{\lambda_{\min}(\tilde{W}) t} -1 \Big) \leq 1$ (observe that $\lambda_{\min}(\tilde{W}) \ge 0$ holds as $-\tilde{W} \preceq 0$ is satisfied.). The inequality~\eqref{eq:final_ineq} results in the property in Proposition~\ref{thm:main}. This completes the proof. 
\end{proof}

\begin{remark}
  It is noteworthy that in \cite{poznyak1995nonlinear}, the identification error analysis of the same DNN is considered. However, it imposes more restrictive conditions (or assumptions) on the DNN; for instance, in the current paper, we assume that $ x_d \in C^1_{n_d}(\mathbb{R}_{\geq 0})$ holds, so that the condition: $f(x_d, \ell({x_d}))$ $< +\infty$ is milder since $f(x_d, \ell({x_d})) < +\infty$ for all $x_d \in \mathbb{R}^{n_d}$ greatly limits the scope of the operator $f$, while $ x_d \in C^1_{n_d}(\mathbb{R}_{\geq 0})$ is also more consistent with the physical properties of the DNN, compared with $x_d \in \mathbb{R}^{n_d}$.
\end{remark}

\section{Case Study} \label{sec:Example}
In this section, we conduct system identification simulations for the power system modeled by the NDAE~\eqref{eq:main_power_sys}. In particular, we consider the IEEE 9-bus test system which is most commonly used in various power system studies \cite{nadeem2022dynamic}. To ensure accurate and efficient identifications, we use the numerical simulation environment composed of Windows $11$, MATLAB R$2023$a, and a $12$th Gen Intel Core i$9$-$12900$H processor with $64$GB RAM. 

In the simulations, we use the MATLAB index $1$ DAE solver \texttt{ode15i} to simulate NDAE models, which is configured with the following parameter values: relative tolerance = $10^{-5}$, absolute tolerance = $10^{-6}$, and a maximum step size of $10^{-3}$. 

\begin{figure}[htb!]
  \centering
  \hspace{0cm}
  \subfigure[The trajectory of relative error $e_a^r(t)$.]{\label{fig:fnn}\includegraphics[width=0.472\textwidth]{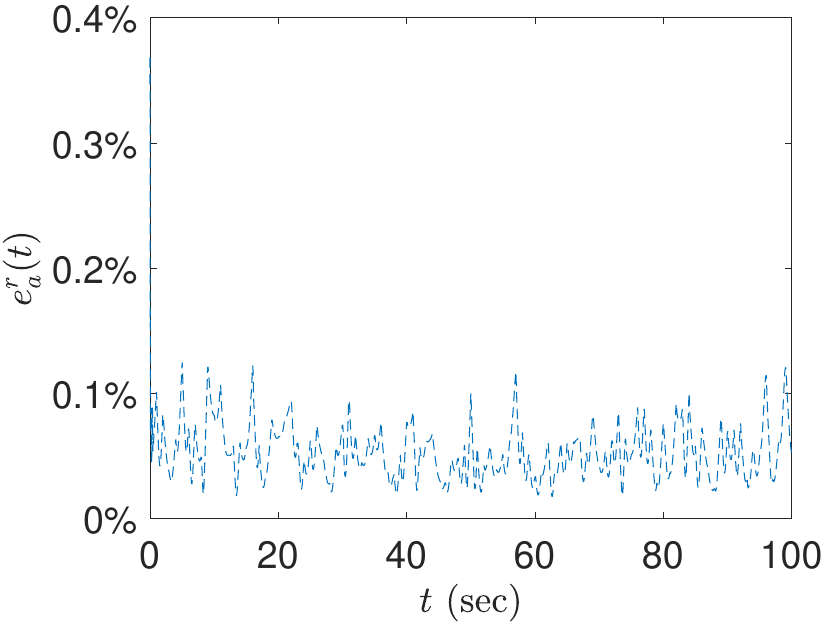}} \hspace{0cm}  \subfigure[The trajectory of relative error $e_d^r(t)$.]{\label{fig:dnn}\includegraphics[width=0.472\textwidth]{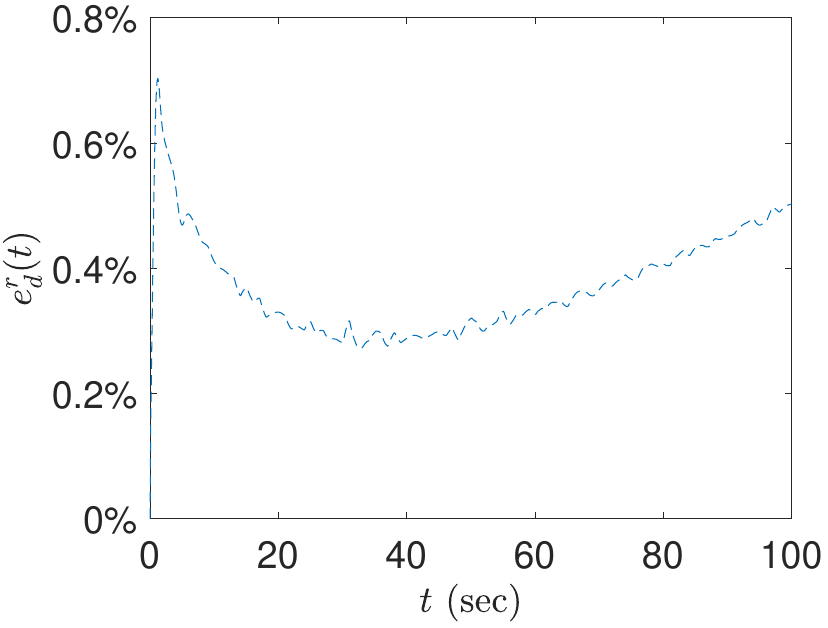}}  \\ \hspace{0.55cm}
  \subfigure[The comparison of $x_{d,4}$ and $\hat{x}_{d,4}$ (in the case of minimum identification error).]{\label{fig:DNN_4}\includegraphics[width=0.442\textwidth]{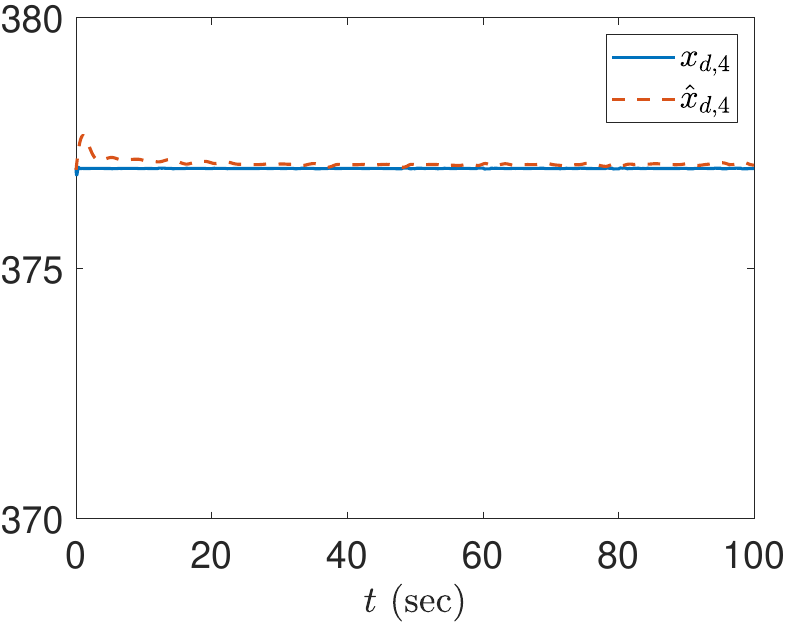}} \hspace{0.55cm}
 \subfigure[The comparison of $x_{d,5}$ and $\hat{x}_{d,5}$ (in the case of maximum identification error).]{\label{fig:DNN_5}\includegraphics[width=0.442\textwidth]{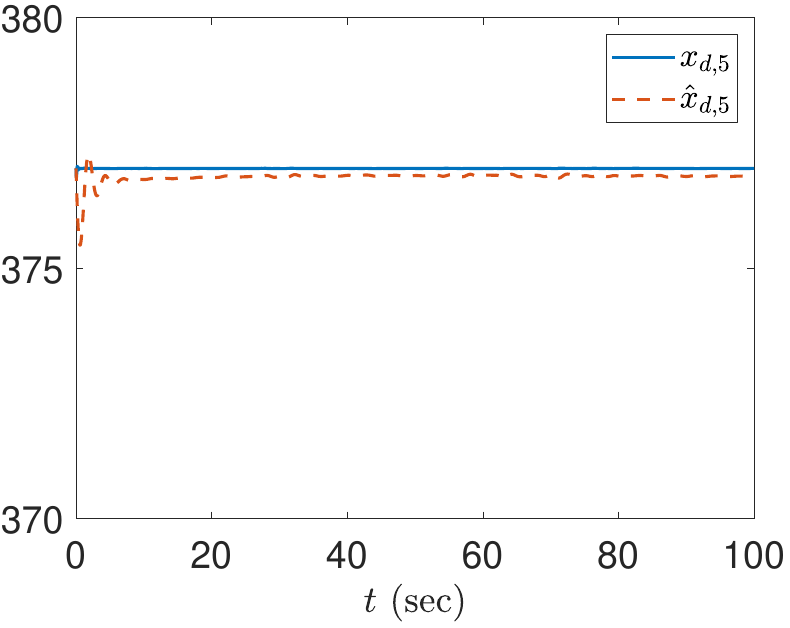}}
\caption{The identification results of the proposed NN-based system identification technique.}
\label{fig:I}
\end{figure}

In the model of~\eqref{eq:main_power_sys}, the variables $x_d,x_a$ are in the spaces $\mathbb{R}^{12}, \mathbb{R}^{24}$, respectively. We randomly select the initial conditions for~\eqref{eq:main_power_sys} and measure the output $y^1$ as defined in~\eqref{eq:output_data}, under a given input $u$ (here, we assume that none of the parameters of~\eqref{eq:main_power_sys} is known).  From $y^1$, we then obtain the data set of $\{x_d^{n,i}, x_a^{n,i}\}_{i=0}^\eta$, where the pairs $\{x_d^{n,i},x_a^{n,i} \}$ are again randomly selected and close to each other. 

%Second, after identifying the NDAE power system~\eqref{eq:main_power_sys}, in the running of the system, we only measure part of variables of $y^1$, that is,  \[ y^2 = \mathcal{T}(y^1), \] where the operator $\mathcal{T}$ extracts part of variables of $y^1$. 
First, an FNN is utilized (the relevant MATLAB built-in function is utilized) to approximate the following algebraic expression:
\begin{align*}
x_a & = \hat{\ell}(x_d),  
\end{align*}
where the explicit form of the function $\hat{\ell}$ is to be estimated. The corresponding simulation results of the identification are depicted by Fig.~\ref{fig:fnn}, in which the relative error $e_a^r(t)$ is defined as follows: 
\begin{align*}
e_a^r(t)~(\%) := 100 \times \frac{\| x_a(t) - \hat{x}_a(t) \| }{ \| x_a(t) \|} \Big(\text{respectively for $x_d$:} \; e_d^r(t)~(\%) := 100 \times \frac{\| x_d(t) - \hat{x}_d(t) \| }{ \| x_d(t) \|} \Big),  
\end{align*}
where $\hat{x}_a$ represents the predicted value of $x_a$ in the identified model of \eqref{main_DAE_system_2}.  

The second step is to identify the dynamic Eq.~\eqref{main_DAE_system_1}. Similarly, $\hat{x}_d$ represents the predicted value of $x_d$ in the identified model of \eqref{main_DAE_system_1}. In this experiment, we have the following dynamic equation:
\begin{align*}
\dot{x}_d &= A_d x_d + C_d f(x_d,\hat{\ell}(x_d)) +  B u + hw_0,
\end{align*}
to be identified. To that end, we use the proposed NN structure (see {Fig.~\ref{Fig:NN_architecture}, in which the output of the FNN is delivered to the inputs of the DNN). We implement the identification of~\eqref{main_DAE_system_1} and obtain the identification errors. We visualize the trajectory of relative error $e^r_d$ in Fig.~\ref{fig:dnn}. It illustrates that the identification error remains small. Furthermore, for brevity, we only visualize the minimum and maximum identification errors, which are $4^{\text{th}}$ and $5^{\text{th}}$ variables of $x_d$ (denoted by $x_{d,4}$ and $x_{d,5}$), respectively. The trajectories of $x_{d,4}$, $\hat{x}_{d,4}$ are shown in Fig.~\ref{fig:DNN_4}, by which we see that the error between the true output $x_{d,4}$ and its predicted value $\hat{x}_{d,4}$ is also small. Moreover, the trajectories of $x_{d,5}$, $\hat{x}_{d,5}$ are presented in Fig.~\ref{fig:DNN_5}, and one can realize that although the corresponding error is larger, it still illustrates the good performance of the proposed DNN. The above simulation results corroborate that our proposed NN scheme has satisfactory performance in the identification of the NDAE-modeled power system~\eqref{eq:main_power_sys}.

% To study the performance of the MOR methodologies in more detail, in the sequel, the \emph{root mean square error (RMSE)} \eqref{eq:RMSE} between actual and recovered states is also used, which can be expressed as follows:   \begin{gather} \label{eq:RMSE} \text{RMSE} := \sum_{i=1}^n  \sqrt{\frac{1}{T} \sum_{t=0}^T (x_i(t) - \hat{x}_i(t))^2}, \end{gather} where $\hat{x} = Vz$ is the recovered state vector. 

\section{Conclusion} \label{sec:conclusion}
This paper introduces an NN framework designed to effectively identify and simulate solution trajectories of NDAEs, such as those characterizing the dynamics of power networks known for exhibiting infinite stiffness. The proposed framework capitalizes on the synergy between IRK time-stepping schemes tailored for NDAEs and NNs (with a DNN). The framework compels an NN to adhere to the algebraic equation of NDAEs as hard constraints and its DNN component is suitable for simulating the ODE-modeled dynamic equation of NDAEs, using a penalty-based method. The paper showcases the efficacy of the proposed NN via the identification and simulation of solution trajectories for the considered DAE-modeled power system. Future work will consider the effect of disturbances that
are present in models, identifying power systems with partially unknown parameters and extending the proposed results to advanced power networks.

\acks{This work is supported by the National Science Foundation under
Grants 2152450 and 2151571.}

\bibliography{ieeeconf.bib}

\end{document}